\documentclass[a4paper]{article}

\usepackage[british]{babel}
\usepackage[T1]{fontenc}
\usepackage[ansinew]{inputenc}
\usepackage{array}
\usepackage{a4}
\usepackage{a4wide}
\usepackage{Common/theorems}
\usepackage{Common/prooftree}

\usepackage{amsmath,amssymb,stmaryrd}

\long\def\ignore#1{\relax}

\newcommand\struto[1][15pt]{{\raise #1 \hbox{\strut}}}%
\newcommand\strutb[1][15pt]{{\raise-#1 \hbox{\strut}}}%

\newcommand\upline{\hline\struto[12pt]}
\newcommand\midline[1][4]{\\[+ #1pt]\hline\struto[#1pt]}

\newcommand\downline[1][12]{\\[+ #1pt]\hline}

\newcommand\pushright[1]{{              
    \parfillskip=0pt            
    \widowpenalty=1000000000         
    \displaywidowpenalty=10000  
    \finalhyphendemerits=0      
    \leavevmode                 
    \unskip                     
    \nobreak                    
    \hfil                       
    \penalty50                  
    \hskip.2em                  
    \null                       
    \hfill                      
    {#1}                        
    \par}}                      

\makeatletter

\renewcommand\[[1][1]{\par\removelastskip\vskip#1pt\vbox\bgroup\null\hfill$}
\renewcommand\][1][3]{$\hfill\null\egroup\par\removelastskip\vskip#1pt\noindent}

\newenvironment{centre}{\par\vbox\bgroup\null\hfill}{\hfill\null\egroup\par\vskip2pt}

\newcommand{\eqdef}{:=\ }
\newcommand{\recdef}{::=\ }

\newcommand{\Gam}{\Gamma}

\newcommand{\Del}{\Delta}

\newcommand\LKF{\textsf{LKF}}

\newcommand\mathFomega{F_\omega}
\newcommand\Fomega{\ifmmode\mathFomega\else$\mathFomega$\fi}
\newcommand\mathFomegaC{F_\omega^{\mathcal C}}
\newcommand\FomegaC{\ifmmode\mathFomegaC\else$\mathFomegaC$\fi}
\newcommand\mathDNE{\mathrm{DNE}}
\newcommand\DNE{\ifmmode\mathDNE\else$\mathDNE$\fi}

\def\url#1#2{\texttt{#1}}

\newcommand\monthdisplay[1]{}

\ignore{

}

\newcommand\FV[2][{}]{\textsf{FV}_{#1}(#2)}

\newcommand{\subst}[3]{ \left\{{}^{#3}\hspace{-6pt}\diagup\hspace{-2pt}_{#2} \right\}\hspace{-1pt} #1 }

\newcommand{\seqg}[3]{\mbox{$\ {#1}_{#2}^{#3}\ $}}

\newcommand{\seqf}[2][]{\seqg{\vdash}{#1}{#2}}

\newcommand{\seq}[1][]{\seqf[#1]{}}

\newcommand\Seq[3][]{#2\seq[{#1}] #3}

\newcommand\DerOSPos[3][]{\Seq[#1]{#2}{[#3]}}
\newcommand\DerOSNeg[3][]{\Seq[#1]{#2}{#3}}

\newcommand\DerNeg[4]{{#1}   \seqf{#4}   {#2}}

\newcommand\DerPosTh[5][{\mathcal T}]{{#2}   \seqf[#1]{#5}   {[#3]}}
\newcommand\DerNegTh[5][{\mathcal T}]{{#2}   \seqf[#1]{#5}   {#3}}

\newcommand\cut{\textsf{cut}}

\newcommand\daggerL{\raise3pt\hbox{\rotatebox{-40}{$\dagger$}}}
\newcommand\daggerR{\raise0pt\hbox{\rotatebox{40}{$\dagger$}}}

\newcommand\andP{{\wedge^+}}
\newcommand\andN{{\wedge^-}}
\newcommand\orP{{\vee^+}}
\newcommand\orN{{\vee^-}}

\newcommand\EX[2]{\exists #1 #2}
\newcommand\FA[2]{\forall #1 #2}

\newcommand{\non}[1]{{#1}^{\perp}}

\newcommand\Theory[2][\mathcal T]{{#1}(#2)}

\newcommand\LKTh[1][\mathcal T]{\textsf{LK}($#1$)}
\newcommand\LKThp[1][\mathcal T]{\textsf{LK}$^p$($#1$)}

\newcommand\DPLLTh{\textsf{DPLL}($\mathcal T$)}

\newcommand\atmCtxt[1]{\textsf{lit} (#1)}

\begin{document}
\title{A sequent calculus with procedure calls}

\author{Mahfuza Farooque${}^1$, Stéphane Lengrand${}^{1,2}$\\[15pt]
  ${}^1$ CNRS\\
  ${}^2$ Ecole Polytechnique\\
  Project PSI: ``Proof Search control in Interaction with domain-specific methods''\\
  ANR-09-JCJC-0006
}

\maketitle
\abstract{
  In this paper, we extend the sequent calculus \LKF~\cite{liang09tcs} into a calculus \LKTh, allowing calls to a decision procedure. We prove cut-elimination of \LKTh.
}

\tableofcontents

\newpage

\section{The sequent calculus \LKTh}

The sequent calculus \LKTh\ manipulates the formulae of first-order logic, with the specificity that every predicate symbol is classified as either positive or negative, and boolean connectives come in two versions: positive and negative.

\begin{definition}[Formulae]\strut
  \emph{Literals} are predicates (a predicate symbol applied to a list of first-order terms) or negations of predicates. Literals are equipped with the obvious involutive negation, and the negation of a literal $l$ is denoted $\non l$.

  Let $\mathcal P$ be the set of literals that are either predicates with positive predicate symbols, or negations of predicates with negative predicate symbols.

  \[
  \begin{array}{lll}
    \mbox{Positive formulae }&P&\recdef p\mid A\andP B \mid A\orP B \mid \EX x A\\
    \mbox{Negative formulae }&N&\recdef \non p\mid A\andN B \mid A\orN B \mid \FA x A\\
    \mbox{Formulae }&A,B&\recdef P\mid N
  \end{array}
  \]
where $p$ ranges over $\mathcal P$.
\end{definition}

\begin{definition}[Negation]Negation is extended from literals to all formulae:
  \[
  \begin{array}{|ll|ll|}
    \hline
    \non {(p)}&\eqdef \non p& \non {(\non p)}&\eqdef p  \\
    \non{(A\andP B)}&\eqdef \non A \orN\non B&\non{(A\andN B)}&\eqdef \non A \orP\non B    \\
    \non{(A\orP B)}&\eqdef \non A \andN\non B&\non{(A\orN B)}&\eqdef \non A \andP\non B    \\
    \non{(\EX x A)}&\eqdef \FA x {\non A}& \non{(\FA x A)}&\eqdef \EX x {\non A}   \\
    \hline
  \end{array}
  \]
\end{definition}

\begin{definition}[\LKTh]

  The sequent calculus \LKThp\ manipulates two kinds of sequents:
  \begin{centre}
    \begin{tabular}{ l l }
      \mbox{Focused sequents }&
      $\DerPosTh \Gamma P {}{}$ \\
      \mbox{Unfocused sequents }&
      $\DerNegTh \Gamma {\Del}{}{}$
    \end{tabular}
  \end{centre}
  where $\Gam$ is a multiset of negative formulae and positive literals, $\Delta$ is a multiset of formulae, and \emph{P} is said to be in the \emph{focus} of the (focused) sequent.
By $\atmCtxt\Gam$ we denote the sub-multiset of $\Gam$ consisting of its 
 literals.

  The rules of \LKThp, given in Figure~\ref{fig:LKTh}, are of three kinds: synchronous rules, asynchronous rules, and structural rules. These correspond to three alternating phases in the proof-search process that is described by the rules.

  If $S$ is a set of literals, $\Theory S$ is the call to the decision procedure on the conjunction of all literals of $S$. It holds if the procedure returns \textsf{UNSAT}.
\end{definition}

\begin{figure}
  $$
  \begin{array}{|c|}
    \upline
    \mbox{\textsf{Synchronous rules}}
    \hfill\strut\\
    \infer{\DerOSPos{\Gamma}{A\andP B}}
    {\DerOSPos{\Gamma}{A} \qquad \DerOSPos{\Gamma}{B}}
    \qquad
    \infer{\DerOSPos{\Gamma}{A_1\orP A_2}}
    {\DerOSPos{\Gamma}{A_i}}
    \qquad
    \infer{\DerOSPos{\Gamma}{\EX x A}}
    {\DerOSPos{\Gamma}{\subst A x t}}
    \\\\
    \infer[p \mbox{ positive literal}]
    {\DerOSPos{\Gamma,p}{p}}
    {\strut}
    \qquad
    \infer[p \mbox{ positive literal}]{\DerOSPos{\Gamma}{p}}
    {\Theory{\atmCtxt\Gam,\non p}}
    \\\\
    \infer[N\mbox{ negative}]{\DerOSPos {\Gam} {N} }{\DerOSNeg {\Gam} {N} }
    \midline
    \mbox{\textsf{Aynchronous rules}}\hfill\strut\\
    \infer{\DerOSNeg{\Gamma}{A\andN B,\Delta}}
    {\DerOSNeg{\Gamma}{A,\Delta} \qquad \DerOSNeg
      {\Gamma}{B,\Delta}}
    \qquad
    \infer{\DerOSNeg {\Gamma} {A_1\orN A_2,\Delta} }
    {\DerOSNeg {\Gamma} {A_1,A_2,\Delta} }
    \qquad
    \infer[x\notin\FV{\Gam,\Delta}]{\DerOSNeg
      {\Gamma}{(\FA x A),\Delta}}
    {\DerOSNeg {\Gamma} {A,\Delta}}\\\\
    \infer[A\mbox{ positive or literal}]{\DerOSNeg \Gam {A,\Delta}} {\DerOSNeg {\Gam,\non A} {\Delta}}
    \midline
    \mbox{\textsf{Structural rules}}\hfill\strut\\
    \infer[\begin{array}l P \mbox{ positive}
    \end{array}]{\DerOSNeg {\Gam,\non P} {}} {\DerOSPos {\Gam,\non P} {P} }
    \qquad
    \infer{\DerOSNeg {\Gam} {}}{\Theory {\atmCtxt\Gam}}
    \downline
  \end{array}
  $$
  \caption{System \LKTh}
  \label{fig:LKTh}
\end{figure}

\section{Admissible rules}

\begin{definition}[Assumptions on the procedure]\strut

  We assume that the procedure calls satisfy the following properties:
  \begin{itemize}
  \item[Weakening] If $\Theory S$ then $\Theory{S,S'}$.
  \item[Contraction] If $\Theory{S,A,A}$ then $\Theory{S,A}$.
  \item[Instantiation] If $\Theory S$ then $\Theory{\subst S x t}$.
  \item[Consistency] If $\Theory {S,p}$ and $\Theory {S,\non p}$ then $\Theory{S}$.
  \end{itemize}
  where $S$ is a set of literals.
\end{definition}

\begin{lemma}[Admissibility of weakening and contraction]\strut

  The following rules are admissible in \LKTh. 
  \[\begin{array}{c}
    \infer{\DerOSPos {\Gam,A}{B}}{\DerOSPos \Gam B}
    \qquad
    \infer{\DerOSNeg {\Gam,A}{\Delta}}{\DerOSNeg \Gam \Del}\\\\
    \infer{\DerOSPos {\Gam,A}{B}}{\DerOSPos {\Gam,A,A} B}
    \qquad
    \infer{\DerOSNeg {\Gam,A}{\Delta}}{\DerOSNeg {\Gam,A,A} \Del}
  \end{array}
  \]
\end{lemma}
\begin{proof}
  By induction on the derivation of the premiss.
\end{proof}

\begin{lemma}[Admissibility of instantiation]  The following rules are admissible in \LKTh. 
  \[
  \infer{\DerOSPos {\subst \Gam x t}{\subst B x t}}{\DerOSPos \Gam B}
  \qquad
  \infer{\DerOSNeg {\subst \Gam x t}{\subst \Del x t}}{\DerOSNeg \Gam \Del}
  \]
\end{lemma}
\begin{proof}
  By induction on the derivation of the premiss.
\end{proof}

\section{Invertibility of the asynchronous phase}

\begin{lemma}[Invertibility of asynchronous rules]\strut

  All asynchronous rules are invertible in \LKTh.
\end{lemma}
\begin{proof}
  By induction on the derivation proving the conclusion of the asynchronous rule considered.
  \begin{itemize}
  \item Inversion of $A \andN B$: by case analysis on the last rule actually used
    \begin{itemize}
    \item $\infer{\DerOSNeg{\Gamma}{A\andN B,C\andN D,\Delta'}} {\DerOSNeg{\Gamma}{A\andN B,C,\Delta'} \quad \DerOSNeg{\Gamma}{A\andN B,D,\Delta'}}$

      By induction hypothesis we get\\[5pt]
      $\infer{\DerOSNeg{\Gamma}{A,C\andN D,\Delta'}} {\DerOSNeg{\Gamma}{A,C,\Delta'} \qquad \DerOSNeg{\Gamma}{A,D,\Delta'}}$
      and
      $\infer{\DerOSNeg{\Gamma}{B,C\andN D,\Delta'}} {\DerOSNeg{\Gamma}{B,C,\Delta'} \quad \DerOSNeg{\Gamma}{B,D,\Delta'}}$

    \item
      $\infer{\DerOSNeg{\Gamma}{A\andN B,C\orN D,\Delta'}} {\DerOSNeg{\Gamma}{A\andN B,C, D,\Delta'}}$

      By induction hypothesis we get $\infer{\DerOSNeg{\Gamma}{A,C\orN D,\Delta'}} {\DerOSNeg{\Gamma}{A,C, D,\Delta'}}$ and
      $\infer{\DerOSNeg{\Gamma}{B,C\orN D,\Delta'}} {\DerOSNeg{\Gamma}{B,C, D,\Delta'}}$

    \item
      $\infer[x\notin\FV{\Gam,\Delta',A\andN B}]{\DerOSNeg{\Gamma}{A\andN B,(\FA x C),\Delta'}} {\DerOSNeg{\Gamma}{A\andN B,C,\Delta'}}$

      By induction hypothesis we get\\
      $\infer[x\notin\FV{\Gam,\Delta',A}]{\DerOSNeg{\Gamma}{A,(\FA x C),\Delta'}} {\DerOSNeg{\Gamma}{A,C,\Delta'}}$ and
      $\infer[x\notin\FV{\Gam,\Delta',B}]{\DerOSNeg{\Gamma}{B,(\FA x C),\Delta'}} {\DerOSNeg{\Gamma}{B,C,\Delta'}}$

    \item
      $\infer[C\mbox{ positive or literal}]{\DerOSNeg{\Gamma}{A\andN B,C,\Delta'}} {\DerOSNeg{\Gamma,\non C}{A\andN B,\Delta'}}$  

      By induction hypothesis we get\\
      $\infer[C\mbox{ positive or literal}]{\DerOSNeg{\Gamma}{A,C,\Delta'}} {\DerOSNeg{\Gamma,\non C}{A,\Delta'}}$   and
      $\infer[C\mbox{ positive or literal}]{\DerOSNeg{\Gamma}{ B,C,\Delta'}} {\DerOSNeg{\Gamma,\non C}{B,\Delta'}} $ 

    \end{itemize}

  \item Inversion of $A \orN B$   
    \begin{itemize}
    \item $\infer{\DerOSNeg{\Gamma}{A\orN B,C\andN D,\Delta'}} {\DerOSNeg{\Gamma}{A\orN B,C,\Delta'} \quad \DerOSNeg{\Gamma}{A\orN B,D,\Delta'}}$   

      By induction hypothesis we get  $ \infer{\DerOSNeg{\Gamma}{A,B,C\andN D,\Delta'}} {\DerOSNeg{\Gamma}{A,B,C,\Delta'}\quad\DerOSNeg{\Gamma}{A,B,D,\Delta'}}$ 

    \item
      $\infer{\DerOSNeg{\Gamma}{A\orN B,C\orN D,\Delta'}} {\DerOSNeg{\Gamma}{A\orN B,C, D,\Delta'}}$  

      By induction hypothesis we get $\infer{\DerOSNeg{\Gamma}{A,B,C\orN D,\Delta'}} {\DerOSNeg{\Gamma}{A,B,C, D,\Delta'}}$ 

    \item
      $\infer[x\notin\FV{\Gam,\Delta'}]{\DerOSNeg{\Gamma}{A\orN B,(\FA x C),\Delta'}} {\DerOSNeg{\Gamma}{A\orN B,C,\Delta'}}$   

      By induction hypothesis we get $\infer[x\notin\FV{\Gam,\Delta'}]{\DerOSNeg{\Gamma}{A,B,(\FA x C),\Delta'}} {\DerOSNeg{\Gamma}{A,B,C,\Delta'}}$  

    \item
      $\infer[C\mbox{ positive or literal}]{\DerOSNeg{\Gamma}{A\orN B,C,\Delta'}} {\DerOSNeg{\Gamma,\non C}{A\orN B,\Delta'}}$  

      By induction hypothesis we get  $\infer[C\mbox{ positive or literal}]{\DerOSNeg{\Gamma}{A,B,C,\Delta'}} {\DerOSNeg{\Gamma,\non C}{A,B,\Delta'}}$   
    \end{itemize}

  \item Inversion of $\FA x A$  

    \begin{itemize}
    \item $\infer{\DerOSNeg{\Gamma}{(\FA x A),C\andN D,\Delta'}} {\DerOSNeg{\Gamma}{(\FA x A),C,\Delta'} \quad \DerOSNeg{\Gamma}{(\FA x A),D,\Delta'}}$ 

      By induction hypothesis  we get $\infer[x\notin\FV{\Gam,\Delta'}]{\DerOSNeg{\Gamma}{A,C\andN D,\Delta'}} {\DerOSNeg{\Gamma}{A,C,\Delta'}}$  and ${\DerOSNeg{\Gamma}{A,D,\Delta'}}$ 
      
    \item
      $\infer{\DerOSNeg{\Gamma}{(\FA x A),C\orN D,\Delta'}} {\DerOSNeg{\Gamma}{(\FA x A),C, D,\Delta'}}$  

      By induction hypothesis we get $\infer{\DerOSNeg{\Gamma}{A,C\orN D,\Delta'}} {\DerOSNeg{\Gamma}{A,C, D,\Delta'}}$ 

    \item  
      $\infer[x\notin\FV{\Gam,\Delta'}]{\DerOSNeg{\Gamma}{(\FA x A),(\FA x D),\Delta'}} {\DerOSNeg{\Gamma}{(\FA x A),D,\Delta'}}$  

      By induction hypothesis we get $\infer[x\notin\FV{\Gam,\Delta'}]{\DerOSNeg{\Gamma}{A,(\FA x C),\Delta'}} {\DerOSNeg{\Gamma}{A,C,\Delta'}} $

    \item
      $\infer[C\mbox{ positive or literal}]{\DerOSNeg{\Gamma}{(\FA x A),C,\Delta'}} {\DerOSNeg{\Gamma,\non C}{(\FA x A),\Delta'}}$  

      By induction hypothesis we get $\infer[C\mbox{ positive or literal}]{\DerOSNeg{\Gamma}{A,C,\Delta'}} {\DerOSNeg{\Gamma,\non C}{A,\Delta'}}$     
    \end{itemize}

  \item Inversion of literals and positive formulae ($A$)
    \begin{itemize}
    \item $\infer{\DerOSNeg{\Gamma}{A,C\andN D,\Delta'}} {\DerOSNeg{\Gamma}{A,C,\Delta'} \quad \DerOSNeg{\Gamma}{A,D,\Delta'}}$  

      By induction hypothesis we get $\infer{\DerOSNeg{\Gamma,\non A}{C\andN D,\Delta'}} {\DerOSNeg{\Gamma,\non A}{C,\Delta'} \quad \DerOSNeg{\Gamma,\non A}{D,\Delta'}}$ 

    \item
      $\infer{\DerOSNeg{\Gamma}{A,C\orN D,\Delta'}} {\DerOSNeg{\Gamma}{A,C, D,\Delta'}}$  

      By induction hypothesis $\infer{\DerOSNeg{\Gamma,\non A}{C\orN D,\Delta'}} {\DerOSNeg{\Gamma,\non A}{C, D,\Delta'}}$ 

    \item
      $\infer[x\notin\FV{\Gam,\Delta'}]{\DerOSNeg{\Gamma}{A,(\FA x D),\Delta'}} {\DerOSNeg{\Gamma}{A,D,\Delta'}}$  \\
      By induction hypothesis we get $\infer[x\notin\FV{\Gam,\Delta'}]{\DerOSNeg{\Gamma,\non A}{(\FA x C),\Delta'}} {\DerOSNeg{\Gamma,\non A}{C,\Delta'}}$  

    \item
      $\infer[B\mbox{ positive or literal}]{\DerOSNeg{\Gamma}{A,B,\Delta'}} {\DerOSNeg{\Gamma,\non B}{A,\Delta'}}$

      By induction hypothesis we get $\infer[B\mbox{ positive or literal}]{\DerOSNeg{\Gamma,\non A}{B,\Delta'}} {\DerOSNeg{\Gamma,\non A,\non B}{\Delta'}}$ 
    \end{itemize}
  \end{itemize}
\end{proof}

\section{Cut-elimination}

\begin{theorem}[$\cut_1$ and $\cut_2$]
  The following rules are admissible in \LKTh. 
  \[\begin{array}{c}
    \infer[\cut_1]{\DerOSNeg {\Gam}{\Delta}}
    {\Theory{\atmCtxt\Gam,\non p}\quad \DerOSNeg {\Gam,p}{\Delta}}
    \qquad
    \infer[\cut_2]{\DerOSPos {\Gam}{B}}
    {\Theory{\atmCtxt\Gam,\non p}\quad \DerOSPos {\Gam,p}{B}}
  \end{array}
  \]
\end{theorem}
\begin{proof}
  By simultaneous induction on the derivation of the right premiss.

  We reduce $\cut_ 8$ by case analysis on the last rule used to prove the right premiss.

  \[
  \infer[\cut_1]{\DerOSNeg {\Gam}{B\andN C, \Delta}}
    {
      \Theory {\atmCtxt\Gam, \non p}
      \quad 
      \infer{
        \DerOSNeg {\Gam,p}{B\andN C,\Delta}
      }
      {
        \DerOSNeg {\Gam,p}{B,\Delta}\quad\DerOSNeg {\Gam,p}{C, \Delta}
      }
    }
    \]
    reduces to
    \[
    \infer{
      \DerOSNeg {\Gam}{B\andN C,\Delta}
    }
    {
      \infer[\cut_1]{\DerOSNeg {\Gam}{B,\Delta}}{\Theory {\atmCtxt\Gam, \non p} \quad \DerOSNeg {\Gam,p}{B,\Delta}}
      \quad
      \infer[\cut_1]{\DerOSNeg {\Gam}{C,\Delta}}{\Theory {\atmCtxt\Gam, \non p} \quad \DerOSNeg {\Gam,p}{C,\Delta}}
    }
  \]

  \[\begin{array}{c}
    \infer[\cut_1]{\DerOSNeg {\Gam}{B_1\orN B_2,\Delta}}
    {
      \Theory {\atmCtxt\Gam, \non p}
      \quad 
      \infer{
        \DerOSNeg {\Gam,p}{B_1\orN B_2,\Delta}
      }
      { \DerOSNeg {\Gam,p}{B_1,B_2,\Delta}    }
    }
    \qquad\mbox{reduces to} \qquad
    \infer { \DerOSNeg {\Gam}{{B_1\orN B_2,\Delta}} }
    {  \infer[\cut_1]{\DerOSNeg {\Gam}{B_1, B_2,\Delta}} {
        \Theory {\atmCtxt\Gam, \non p} \quad \DerOSNeg {\Gam,p}{B_1, B_2,\Delta}
      }} 
  \end{array}
  \]

  \[\begin{array}{c}
    \infer[\cut_1]{\DerOSNeg {\Gam}{\FA x B,\Delta}}
    {
      \Theory {\atmCtxt\Gam, \non p}
      \quad 
      \infer{
        \DerOSNeg {\Gam,p}{\FA x B,\Delta}
      }
      { \DerOSNeg {\Gam,p}{B,\Delta}    }
    }
    \qquad\mbox{reduces to} \qquad
    \infer{ \DerOSNeg {\Gam}{{\FA x B,\Delta}} }
    {  \infer[\cut_1] {\DerOSNeg {\Gam}{B,\Delta}} {
        \Theory {\atmCtxt\Gam, \non p} \quad \DerOSNeg {\Gam,p}{B,\Delta}
      }} 
  \end{array}
  \]

  \[\begin{array}{c}
    \infer[\cut_1]{\DerOSNeg {\Gam}{B,\Delta}}
    {
      \Theory {\atmCtxt\Gam, \non p}
      \quad 
      \infer{
        \DerOSNeg {\Gam,p}{B,\Delta}
      }
      { \DerOSNeg {\Gam,p,\non B}{\Delta}    }
    }
    \qquad\mbox{reduces to} \qquad
    \infer { \DerOSNeg {\Gam}{{B,\Delta}} }
    {  \infer[\cut_1]{\DerOSNeg {\Gam, \non B}{\Delta}} {
        \Theory {\atmCtxt{\Gam,\non B}, \non p} \quad \DerOSNeg {\Gam,p,\non B}{\Delta}
      }} 
  \end{array}
  \]
  We have $\Theory {\atmCtxt\Gam, \non p,\non B}$ as we assume the procedure to satisfy weakening.


  If $\non P\in(\Gamma,p)$,
  \[\begin{array}{c}
    \infer[\cut_1]{\DerOSNeg {\Gam}{}}
    {
      \Theory {\atmCtxt\Gam, \non p}
      \quad 
      \infer{
        \DerOSNeg {\Gam,p}{}
      }
      { \DerOSPos {\Gam,p}{P}   }
    }
    \qquad\mbox{reduces to} \qquad
    \infer { \DerOSNeg {\Gam}{} }
    {  \infer[\cut_2]{\DerOSPos {\Gam}{P}} {
        \Theory {\atmCtxt\Gam, \non p} \quad \DerOSPos {\Gam, \non p}{P}
      }} 
  \end{array}
  \]
  as $\non P\in(\Gamma)$.


  \[\begin{array}{c}
    \infer[\cut_1]{\DerOSNeg {\Gam}{}}
    {
      \Theory {\atmCtxt\Gam, \non p}
      \quad 
      \infer{
        \DerOSNeg {\Gam,p}{}
      }
      { \Theory {\atmCtxt\Gam,p}    }
    }
    \qquad\mbox{reduces to} \qquad
    \infer { \DerOSNeg {\Gam}{} }
    {  \Theory {\atmCtxt\Gam}}
  \end{array}
  \]
  using the assumption of consistency.

  We reduce $\cut_2$ again by case analysis on the last rule used to prove the right premiss.

  \[\infer[\cut_2]{\DerOSPos {\Gam}{B\andP C}}
    {
      \Theory {\atmCtxt\Gam, \non p}
      \quad 
      \infer{
        \DerOSPos {\Gam,p}{B\andP C}
      }
      {
        \DerOSPos {\Gam,p}{B}\quad\DerOSPos {\Gam,p}{C}
      }
    }\]
    reduces to
    \[
    \infer{
      \DerOSPos {\Gam}{B\andP C}
    }
    {
      \infer[\cut_2]{\DerOSPos {\Gam}{B}}{\Theory {\atmCtxt\Gam, \non p} \quad \DerOSPos {\Gam,p}{B}}
      \quad
      \infer[\cut_2]{\DerOSPos {\Gam}{C}}{\Theory {\atmCtxt\Gam, \non p} \quad \DerOSPos {\Gam,p}{C}}
    }
  \]

  \[\begin{array}{c}
    \infer[\cut_2]{\DerOSPos {\Gam}{B_1\orP B_2}}
    {
      \Theory {\atmCtxt\Gam, \non p}
      \quad 
      \infer{
        \DerOSPos {\Gam,p}{B_1\orP B_2}
      }
      { \DerOSPos {\Gam,p}{B_i}    }
    }
    \qquad\mbox{reduces to} \qquad
    \infer { \DerOSPos {\Gam}{{B_1\orP B_2}} }
    {  \infer[\cut_2]{\DerOSPos {\Gam}{B_i}} {
        \Theory {\atmCtxt\Gam, \non p} \quad \DerOSPos {\Gam,p}{B_i}
      }} 
  \end{array}
  \]

  \[\begin{array}{c}
    \infer[\cut_2] {\DerOSPos {\Gamma}{\EX x B}}
    {
      \Theory {\atmCtxt\Gam, \non p}
      \quad 
      \infer{
        \DerOSPos {\Gam,p}{\EX x B}
      }
      { \DerOSPos {\Gam,p} {\subst B x t}    }
    }
    \qquad\mbox{reduces to} \qquad
    \infer { \DerOSPos {\Gam}{{\EX x B}} }
    {  \infer[\cut_2]{\DerOSPos {\Gam} {\subst B x t}} {
        \Theory {\atmCtxt\Gam, \non p} \quad \DerOSPos {\Gam,p}{\subst B x t}
      }} 
  \end{array}
  \]

  \[\begin{array}{c}
    \infer[\cut_2]{\DerOSPos {\Gam}{N}}
    {
      \Theory {\atmCtxt\Gam, \non p}
      \quad 
      \infer{
        \DerOSPos {\Gam,p}{N}
      }
      { \DerOSNeg {\Gam,p}{N}    }
    }
    \qquad\mbox{reduces to} \qquad
    \infer { \DerOSPos {\Gam}{{N}} }
    {  \infer[\cut_1]{\DerOSNeg {\Gam}{N}} {
        \Theory {\atmCtxt\Gam, \non p} \quad \DerOSNeg {\Gam,p}{N}
      }} 
  \end{array}
  \]

  If $p'\in \Gamma,p$,
  \[\begin{array}{c}
    \infer[\cut_2]{\DerOSPos {\Gam}{p'}}
    {
      \Theory {\atmCtxt\Gam, \non p}
      \quad 
      \infer{
        \DerOSPos {\Gam,p}{p'}
      }
      {     }
    }
    \qquad\mbox{reduces to} \qquad
    \infer{ \DerOSPos {\Gam}{p'} }{}
    \qquad$if$ \qquad p'\in\Gam\\
    \qquad\mbox{reduces to} \qquad
    \infer{ \DerOSPos {\Gam}{p'} }{\Theory {\atmCtxt\Gam, \non p}}
    \qquad$if$ \qquad p'=p
  \end{array}
  \]

  Finally,
  \[\begin{array}{c}
    \infer[\cut_2]{\DerOSPos {\Gam}{p'}}
    {
      \Theory {\atmCtxt\Gam, \non p}
      \quad 
      \infer{
        \DerOSPos {\Gam,p}{p'}
      }
      {\Theory{\atmCtxt\Gam,p,\non {p'}}}
    }
    \qquad\mbox{reduces to} \qquad
    \infer{ \DerOSPos {\Gam}{p'} }{\Theory {\atmCtxt\Gam, \non {p'}}}
    \end{array}
  \]
since weakening gives $\Theory {\atmCtxt\Gam, \non p,\non {p'}}$ and consistency then gives $\Theory {\atmCtxt\Gam, \non {p'}}$.
\end{proof}

\begin{theorem}[$\cut_3$, $\cut_4$ and $\cut_5$]
  The following rules are admissible in \LKTh. 
  \[\begin{array}{c}
    \infer[\cut_3]{\DerOSNeg {\Gam}{\Del}}
    {\DerOSPos \Gam A\quad \DerOSNeg {\Gam}{\non A,\Del}}
    \\\\
    \infer[\cut_4]{\DerOSNeg {\Gam}{\Del}}
    {\DerOSNeg \Gam {N}\quad \DerOSNeg {\Gam,N}{\Del}}
    \qquad
    \infer[\cut_5]{\DerOSPos {\Gam}{B}}
    {\DerOSNeg \Gam {N}\quad \DerOSPos {\Gam,N}{B}}
  \end{array}
  \]
\end{theorem}

\begin{proof}
  By simultaneous induction on the following lexicographical measure:
  \begin{itemize}
  \item the size of the cut-formula ($A$ or $N$)
  \item the fact that the cut-formula ($A$ or $N$) is positive or negative\\
    (if of equal size, a positive formula is considered smaller than a negative formula)
  \item the height of the derivation of the right premiss
  \end{itemize}

  Weakenings and contractions (as they are admissible in the system) are implicitly used throughout this proof.

  In order to eliminate $\cut_3$, we analyse which rule is used to prove the left premiss. We then use invertibility of the negative phase so that the last rule used in the right premiss is its dual one.

  \[\begin{array}{c}
    \infer[\cut_3]{\DerOSNeg{\Gamma}{\Delta}}
    {\infer{\DerOSPos{\Gamma}{A \andP B}} {{\DerOSPos{\Gamma}{A}} \quad {\DerOSPos{\Gamma}{B}}} \qquad \infer{\DerOSNeg{\Gamma}{A \orN B,\Delta}}{\DerOSNeg{\Gamma}{\non A,\non B},\Delta}}
    \quad \mbox{reduces to} \quad 
    \infer[\cut_3]{\DerOSNeg{\Gamma}{\Delta}}
    {{\DerOSPos{\Gamma}{B}} \quad {\infer[\cut_3]{\DerOSNeg{\Gamma}{\non B,\Delta}} {{\DerOSPos{\Gamma}{A}} \quad {\DerOSNeg{\Gamma}{\non A,\non B},\Delta}}}}
  \end{array}
  \]


  \[\begin{array}{c}
    \infer[\cut_3]{\DerOSNeg{\Gamma}{\Delta}}
    {\infer{\DerOSPos{\Gamma}{A_1 \orP A_2}} {\DerOSPos{\Gamma}{A_i}}  \qquad \infer{\DerOSNeg{\Gamma}{A_1 \andN A_2,\Delta}}{{\DerOSNeg{\Gamma}{\non A_1},\Delta} \quad {\DerOSNeg{\Gamma}{\non A_2},\Delta}}} \quad
    \mbox{reduces to} \qquad
    \infer[\cut_3]{\DerOSNeg{\Gamma}{\Delta}}
    {\DerOSPos{\Gamma}{A_i} \qquad {\DerOSNeg{\Gamma}{\non A_i,\Delta}}}
  \end{array}\]

  \[\begin{array}c
    \infer[\cut_3]{\DerOSNeg{\Gamma}{\Delta}}
    {\infer{\DerOSPos{\Gamma}{\EX x A}} {\DerOSPos{\Gamma}{\subst A x t}} \qquad \infer{\DerOSNeg{\Gamma}{(\FA x \non A),\Delta}}{\DerOSNeg{\Gamma}{\non A},\Delta}} \quad
    \mbox{reduces to} \qquad \infer[\cut_3]{\DerOSNeg{\Gamma}{\Delta}}
    {{\DerOSPos{\Gamma}{\subst A x t}} \qquad \iinfer[ x \notin \FV{\Gamma,\Delta}]{\DerOSNeg{\Gamma}{(\subst {\non A} x t ),\Delta}}{\DerOSNeg{\Gamma}{\non A},\Delta}}
  \end{array}
  \]
  using the admissibility of instantiation.

  \[
  \infer[\cut_3]{\DerOSNeg{\Gamma}{\Delta}}
  {\infer{\DerOSPos{\Gamma}{N}} {\DerOSNeg{\Gamma}{N}} \qquad \infer{\DerOSNeg{\Gamma}{(\non N),\Delta}}{\DerOSNeg{\Gamma, N }{\Delta}}} \quad
  \mbox{reduces to} \qquad \infer[\cut_4]{\DerOSNeg{\Gamma}{\Delta}}
  {{\DerOSNeg{\Gamma}{N}} \qquad {\DerOSNeg{\Gamma, N}{\Delta}}}
  \]
  We will describe below how $\cut_4$  is reduced.

  \[
  \infer[\cut_3]{\DerOSNeg{\Gamma,p}{\Delta}}
  {\infer{\DerOSPos{\Gamma,p}{p}} {} \qquad \infer{\DerOSNeg{\Gamma,p}{(\non p),\Delta}}{\DerOSNeg{\Gamma, p,p }{\Delta}}} \quad
  \mbox{reduces to} \qquad \iinfer{\DerOSNeg{\Gamma,p}{\Delta}}
  {\DerOSNeg{\Gamma, p,p}{\Delta}}
  \]
  using the admissibility of contraction.

  \[
  \infer[\cut_3]{\DerOSNeg{\Gamma}{\Delta}}
  {\infer{\DerOSPos{\Gamma}{p}} {{\Theory{\atmCtxt\Gam,\non p}}} \qquad \infer{\DerOSNeg{\Gamma}{(\non p),\Delta}}{\DerOSNeg{\Gamma, p }{\Delta}}} \quad
  \mbox{reduces to} \qquad \infer[\cut_1]{\DerOSNeg{\Gamma}{\Delta}}
  {{\Theory{\atmCtxt\Gam,\non p}}  \qquad {\DerOSNeg{\Gamma, p}{\Delta}}}
  \]


  In order to reduce $\cut_4$, we analyse which rule is used to prove the right premiss.


  \[\begin{array}{c}
    \infer[\cut_4]{\DerOSNeg {\Gam}{}}
    {\DerOSNeg \Gam {N}\quad \infer{\DerOSNeg {\Gam,N}{}}{\Theory{\atmCtxt\Gam}}}
    \qquad\mbox{reduces to} \qquad
    \infer{\DerOSNeg {\Gam}{}}{\Theory{\atmCtxt\Gam}}
  \end{array}
  \]
  if $N$ is not an literal (hence, it is not passed on to the procedure).

  \[\begin{array}{c}
    \infer[\cut_4]{\DerOSNeg {\Gam}{}}
    {\infer{\DerOSNeg \Gam {\non p}}{\DerOSNeg {\Gam, p}{}}\quad \infer{\DerOSNeg {\Gam,\non p}{}}{\Theory{\atmCtxt\Gam,\non p}}}
    \qquad\mbox{reduces to} \qquad
    \infer[\cut_1]{\DerOSNeg {\Gam}{}}
    {
      \Theory{\atmCtxt\Gam,\non p}
      \qquad
      \DerOSNeg {\Gam, p}{}
    }
  \end{array}
  \]
  if $\non p$ is an literal passed on to the procedure.

  \[\begin{array}{c}
    \infer[\cut_4]{\DerOSNeg {\Gam}{}}
    {\DerOSNeg \Gam {N}\quad \infer{\DerOSNeg {\Gam,N}{}}{\DerOSPos {\Gam,N}{\non N}}}
    \qquad\mbox{reduces to} \qquad
    \infer[\cut_3]{\DerOSNeg {\Gam}{}}{\infer[\cut_5]{\DerOSPos {\Gam}{\non N}}{\DerOSNeg \Gam {N}\quad \DerOSPos {\Gam,N}{\non N}}\qquad \DerOSNeg \Gam {N}}
  \end{array}
  \]

  \[\begin{array}{c}
    \infer[\cut_4]{\DerOSNeg {\Gam,\non P}{}}
    {\DerOSNeg {\Gam,\non P} {N}\quad \infer{\DerOSNeg {\Gam,\non P,N}{}}{\DerOSPos {\Gam,\non P,N}{P}}}
    \qquad\mbox{reduces to} \qquad
    \infer{\DerOSNeg {\Gam,\non P}{}}{\infer[\cut_5]{\DerOSPos {\Gam,\non P}{P}}{\DerOSNeg {\Gam,\non P} {N}\quad \DerOSPos {\Gam,\non P,N}{P}}}
  \end{array}
  \]

  \[
  \infer[\cut_4]{\DerOSNeg {\Gam}{B\andN C,\Delta}}
    {
      \DerOSNeg {\Gam} {N}
      \quad 
      \infer{
        \DerOSNeg {\Gam,N}{B\andN C,\Delta}
      }
      {
        \DerOSNeg {\Gam,N}{B,\Delta}\quad\DerOSNeg {\Gam,N}{C,\Delta}
      }
    }
    \]
    reduces to
    \[
    \infer{
      \DerOSNeg {\Gam}{B\andN C,\Delta}
    }
    {
      \infer[\cut_4]{\DerOSNeg {\Gam}{B,\Delta}}{\DerOSNeg {\Gam} {N}\quad \DerOSNeg {\Gam,N}{B,\Delta}}
      \quad
      \infer[\cut_4]{\DerOSNeg {\Gam}{C,\Delta}}{\DerOSNeg {\Gam} {N}\quad \DerOSNeg {\Gam,N}{C,\Delta}}
    }
  \]

  \[\begin{array}{c}
    \infer[\cut_4]{\DerOSNeg {\Gam}{B\orN C,\Delta}}
    {
      \DerOSNeg {\Gam} {N}
      \quad 
      \infer{
        \DerOSNeg {\Gam,N}{B\orN C,\Delta}
      }
      {
        \DerOSNeg {\Gam,N}{B,C,\Delta}
      }
    }
    \qquad\mbox{reduces to} \qquad
    \infer{
      \DerOSNeg {\Gam}{B\orN C,\Delta}
    }
    {
      \infer[\cut_4]{\DerOSNeg {\Gam}{B,C,\Delta}}{\DerOSNeg {\Gam} {N}\quad \DerOSNeg {\Gam,N}{B,C,\Delta}}
    }
  \end{array}
  \]

  \[\begin{array}{c}
    \infer[\cut_4]{\DerOSNeg {\Gam}{\forall x B,\Delta}}
    {
      \DerOSNeg {\Gam} {N}
      \quad 
      \infer{
        \DerOSNeg {\Gam,N}{\forall x B,\Delta}
      }
      {
        \DerOSNeg {\Gam,N}{B,\Delta}
      }
    }
    \qquad\mbox{reduces to} \qquad
    \infer{
      \DerOSNeg {\Gam}{\forall x B,\Delta}
    }
    {
      \infer[\cut_4]{\DerOSNeg {\Gam}{B,\Delta}}{\DerOSNeg {\Gam} {N}\quad \DerOSNeg {\Gam,N}{B,\Delta}}
    }
  \end{array}
  \]

  \[\begin{array}{c}
    \infer[\cut_4]{\DerOSNeg {\Gam}{B,\Delta}}
    {
      \DerOSNeg {\Gam} {N}
      \quad 
      \infer{
        \DerOSNeg {\Gam,N}{B,\Delta}
      }
      {
        \DerOSNeg {\Gam,N,\non B}{\Delta}
      }
    }
    \qquad\mbox{reduces to} \qquad
    \infer{
      \DerOSNeg {\Gam}{B,\Delta}
    }
    {
      \infer[\cut_4]{\DerOSNeg {\Gam,\non B}{\Delta}}{\DerOSNeg {\Gam,\non B} {N}\quad \DerOSNeg {\Gam,N,\non B}{\Delta}}
    }
  \end{array}
  \]
  using weakening, and if $B$ is positive or a negative literal.

  We have reduced all cases of $\cut_4$; we now reduce the cases for $cut_5$ (again, by case analysis on the last rule used to prove the right premiss).

  \[\begin{array}{c}
    \infer[\cut_5]{\DerOSPos {\Gam}{B\andP C}}
    {
      \DerOSNeg {\Gam} {N}
      \quad 
      \infer{
        \DerOSPos {\Gam,N}{B\andP C}
      }
      {
        \DerOSPos {\Gam,N}{B}\quad\DerOSPos {\Gam,N}{C}
      }
    }
    \qquad\mbox{reduces to} \qquad
    \infer{
      \DerOSPos {\Gam}{B\andP C}
    }
    {
      \infer[\cut_5]{\DerOSPos {\Gam}{B}}{\DerOSNeg {\Gam} {N}\quad \DerOSPos {\Gam,N}{B}}
      \quad
      \infer[\cut_5]{\DerOSPos {\Gam}{C}}{\DerOSNeg {\Gam} {N}\quad \DerOSPos {\Gam,N}{C}}
    }
  \end{array}
  \]


  \[\begin{array}{c}
    \infer[\cut_5]{\DerOSPos {\Gam}{B_1\orP B_2}}
    {
      \DerOSNeg {\Gam} {N}
      \quad 
      \infer{
        \DerOSPos {\Gam,N}{B_1\orP B_2}
      }
      { \DerOSPos {\Gam,N}{B_i}    }
    }
    \qquad\mbox{reduces to} \qquad
    \infer[\cut_5] { \DerOSPos {\Gam}{B_i} }
    {
      {\DerOSNeg {\Gam} {N}\quad \DerOSPos {\Gam,N}{B_i}}
    }
  \end{array}
  \]
  

  \[\begin{array}{c}
    \infer[\cut_5]{\DerOSPos {\Gam}{{\EX x B}}}
    {
      \DerOSNeg {\Gam} {N}
      \quad 
      \infer{
        \DerOSPos {\Gam,N}{{\EX x B}}
      }
      { \DerOSPos {\Gam,N}{{\subst B x t}}    }
    }
    \qquad\mbox{reduces to} \qquad
    \infer { \DerOSPos {\Gam}{{\EX x B}} }
    {   \infer[\cut_5] { \DerOSPos {\Gam}{{\subst B x t}} }
      {\DerOSNeg {\Gam} {N}\quad \DerOSPos {\Gam,N}{{\subst B x t}}
      } }
  \end{array}
  \]


  \[\begin{array}{c}
    \infer[\cut_5]{\DerOSPos {\Gam}{N'}}
    {
      \DerOSNeg {\Gam} {N}
      \quad 
      \infer{
        \DerOSPos {\Gam,N}{N'}
      }
      { \DerOSNeg {\Gam,N}{N'}    }
    }
    \qquad\mbox{reduces to} \qquad
    \infer{\DerOSPos {\Gam}{N'}}
    {
      \infer[\cut_4]{\DerOSNeg {\Gam}{N'}}
      {
        \DerOSNeg {\Gam}{N}
        \quad
        \DerOSNeg {\Gam,N}{N'}
      }
    }
  \end{array}
  \]

  \[\begin{array}{c}
    \infer[\cut_5]{\DerOSPos {\Gam}{p}}
    {
      \DerOSNeg {\Gam} {N}
      \quad 
      \infer{\DerOSPos{\Gamma,N}{p}} {}
    }
    \qquad\mbox{reduces to} \qquad
    \infer { \DerOSPos {\Gam}{p}} {}
    
  \end{array}
  \]

  since $p$ has to be in $\Gamma$.

  \[\begin{array}{c}   

    \infer[\cut_5]{\DerOSPos {\Gam}{p}}
    {
      \DerOSNeg {\Gam} {N}
      \quad 
      \infer{\DerOSPos{\Gamma,N}{p}} {\Theory{\atmCtxt\Gamma, \non p}}
    }
    \qquad\mbox{reduces to} \qquad
    \infer{\DerOSPos {\Gam}{p}}
    {
      \Theory {\atmCtxt\Gam, \non p}
    }
  \end{array}
  \]
  
\end{proof}

\newpage
\begin{theorem}[$\cut_6$, $\cut_7$, $\cut_8$, and $\cut_9$]
  The following rules are admissible in \LKTh. 
  \[
  \begin{array}{c}
    \infer[\cut_6]{\DerOSNeg {\Gam}{\Del}}
    {\DerOSNeg \Gam {N,\Del}\quad \DerOSNeg {\Gam,N}{\Del}}
    \qquad
    \infer[\cut_7] {\DerOSNeg {\Gamma}{\Del}}
    {\DerOSNeg{\Gamma}{A,\Del} \quad \DerOSNeg{\Gamma}{\non A,\Del}}  
\\\\
    \infer[\cut_8] {\DerOSNeg {\Gamma}{\Del}}
    {\DerOSNeg{\Gamma,l}{\Del} \quad \DerOSNeg{\Gamma,\non l}{\Del}}  
\qquad
  \infer[\cut_9]{\DerNeg {\Gam} {\Del}{} {}}
  {
    \DerNeg {\Gam, l_1,\ldots,l_n} {\Del}{} {}
    \quad
    \DerNeg {\Gam, (\non l_1\orN\ldots\orN \non l_n)} {\Del}{} {}
  } 
  \end{array}
  \]
\end{theorem}
\begin{proof}
$\cut_6$ is proved admissible by induction on the multiset $\Del$: the base case is the admissibility of $\cut_4$, and the other cases just require the inversion of the connectives in $\Del$.

For $\cut_7$, we can assume without loss of generality (swapping $A$ and $\non A$) that $A$ is negative.
Applying inversion on $\DerOSNeg{\Gamma}{\non A,\Del}$ gives a proof of $\DerOSNeg{\Gamma,A}{\Del}$, and 
$\cut_7$ is then obtained by $\cut_6$:
\[    \infer[\cut_6] {\DerOSNeg {\Gamma}{\Del}}
    {\DerOSNeg{\Gamma}{A,\Del} \quad 
      \DerOSNeg{\Gamma,A}{\Del}
    }
\]

$\cut_8$ is obtained as follows:
\[    \infer[\cut_7] {\DerOSNeg {\Gamma}{\Del}}
    {\infer
      {\DerOSNeg{\Gamma}{l,\Del}}
      {\DerOSNeg{\Gamma,\non l}{\Del}}
      \quad 
      \infer
      {\DerOSNeg{\Gamma}{\non l,\Del}}
      {\DerOSNeg{\Gamma,l}{\Del}}
    }
\]

$\cut_9$ is obtained as follows:
\[    \infer[\cut_7] {\DerOSNeg {\Gamma}{\Del}}
{\infer
  {\DerOSNeg{\Gamma}{l_1\andP\ldots\andP l_n,\Del}}
  {\DerOSNeg{\Gamma,(\non l_1\orN\ldots\orN\non l_n)}{\Del}}
  \quad
  \infer{\DerOSNeg{\Gamma}{(\non l_1\orN\ldots\orN\non l_n),\Del}}
  {
    \Infer
    {\DerOSNeg{\Gamma}{\non l_1,\ldots,\non l_n,\Del}}
    {\DerOSNeg{\Gamma,l_1,\ldots,l_n}{\Del}}
  }
}
\]

\end{proof}

\section{Conclusion}

It is worth noting that an instance of such a theory is the theory where $\Theory{S}$ holds if and only if there is a literal $p\in S$ such that $\non p\in S$.

We proved the admissibility of $\cut_8$ and $\cut_9$ as they are used to simulate the \DPLLTh\ procedure~\cite{Nieuwenhuis06} as the proof-search mechanism of \LKTh.

Further work will consist in using the cut-admissibility results to:
\begin{itemize}
\item show that changing the polarities of the connectives and predicates that are present in a sequent, does not change the provability of that sequent in \LKTh;
\item prove the completeness of \LKTh\ with respect to the standard notion of provability in first-order logic, working in a particular theory $\mathcal T$ for which we have a (sound and complete) decision procedure;
\item show how the \DPLLTh\ procedure can be simulated in \LKTh\ (with backtracking as well as with backjumping and lemma learning).
\end{itemize}

\bibliographystyle{Common/good}
\bibliography{Common/abbrev-short,Common/Main,Common/crossrefs}

\end{document}